\documentclass[12pt]{amsart}
\usepackage{amssymb}
\usepackage{amsfonts}
\usepackage{amscd}
\usepackage{graphicx}
\usepackage{hyperref}
\hypersetup{pdftitle={Foliations and automorphisms of causal sets},pdfauthor={Ali Bleybel and Abdallah Zaiour}}
\usepackage{microtype}
\newcommand{\sprec}{\precneqq}
\newtheorem*{lemma*}{Lemma}
\newtheorem*{remark*}{Remark}
\newtheorem*{definition*}{Definition}
\newtheorem*{example*}{Example}
\newtheorem*{proposition*}{Proposition}

\swapnumbers

\def\RR{{\mathbf R}}

\def\cC{{\mathcal C}}

\def\cE{{\mathcal E}}

\def\cL{{\mathcal L}}
\def\cM{{\mathcal M}}

\def\cO{{\mathcal O}}

\def\cR{{\mathcal R}}

\def\cY{{\mathcal Y}}

\mathchardef\alphag="7C0B \mathchardef\betag="7C0C
\mathchardef\gammag="7C0D \mathchardef\deltag="7C0E
\mathchardef\varepsilong="7C22 \mathchardef\varphig="7C27
\mathchardef\psig="7C20 \mathchardef\zetag="7C10
\mathchardef\epsilong="7C0F \mathchardef\rhog="7C1A
\mathchardef\taug="7C1C \mathchardef\upsilong="7C1D
\mathchardef\iotag="7C13 \mathchardef\thetag="7C12
\mathchardef\pig="7C19 \mathchardef\sigmag="7C1B
\mathchardef\etag="7C11 \mathchardef\omegag="7C21
\mathchardef\kappag="7C14 \mathchardef\lambdag="7C15
\mathchardef\mug="7C16 \mathchardef\xig="7C18
\mathchardef\chig="7C1F \mathchardef\nug="7C17
\mathchardef\varthetag="7C23 \mathchardef\varpig="7C24
\mathchardef\varrhog="7C25 \mathchardef\varsigmag="7C26
\mathchardef\Omegag="7C0A \mathchardef\Thetag="7C02
\mathchardef\Sigmag="7C06 \mathchardef\Deltag="7C01
\mathchardef\Phig="7C08 \mathchardef\Gammag="7C00
\mathchardef\Psig="7C09 \mathchardef\Lambdag="7C03
\mathchardef\Xig="7C04 \mathchardef\Pig="7C05
\mathchardef\Upsilong="7C07

\newtheorem{theorem}[subsection]{Theorem}
\newtheorem{lem}[subsection]{Lemma}

\newtheorem{prop}[subsection]{Proposition}

\theoremstyle{definition}
\newtheorem{definition}[subsection]{Definition}

\newtheorem{def-prop}[subsubsection]{Definition-Proposition}

\theoremstyle{remark}

\theoremstyle{plain}

\numberwithin{equation}{subsection}

\def\boxit#1#2{\setbox1=\hbox{\kern#1{#2}\kern#1}%
\dimen1=\ht1 \advance\dimen1 by #1 \dimen2=\dp1
\advance\dimen2 by #1
\setbox1=\hbox{\vrule height\dimen1 depth\dimen2\box1\vrule}%
\setbox1=\vbox{\hrule\box1\hrule}%
\advance\dimen1 by .4pt \ht1=\dimen1 \advance\dimen2 by
..4pt \dp1=\dimen2 \box1\relax}

\def\RR{{\mathbf R}}

\def\cC{{\mathcal C}}

\def\cE{{\mathcal E}}

\def\cL{{\mathcal L}}
\def\cM{{\mathcal M}}

\def\cO{{\mathcal O}}

\def\cR{{\mathcal R}}

\def\cY{{\mathcal Y}}

\mathchardef\alphag="7C0B \mathchardef\betag="7C0C
\mathchardef\gammag="7C0D \mathchardef\deltag="7C0E
\mathchardef\varepsilong="7C22 \mathchardef\varphig="7C27
\mathchardef\psig="7C20 \mathchardef\zetag="7C10
\mathchardef\epsilong="7C0F \mathchardef\rhog="7C1A
\mathchardef\taug="7C1C \mathchardef\upsilong="7C1D
\mathchardef\iotag="7C13 \mathchardef\thetag="7C12
\mathchardef\pig="7C19 \mathchardef\sigmag="7C1B
\mathchardef\etag="7C11 \mathchardef\omegag="7C21
\mathchardef\kappag="7C14 \mathchardef\lambdag="7C15
\mathchardef\mug="7C16 \mathchardef\xig="7C18
\mathchardef\chig="7C1F \mathchardef\nug="7C17
\mathchardef\varthetag="7C23 \mathchardef\varpig="7C24
\mathchardef\varrhog="7C25 \mathchardef\varsigmag="7C26
\mathchardef\Omegag="7C0A \mathchardef\Thetag="7C02
\mathchardef\Sigmag="7C06 \mathchardef\Deltag="7C01
\mathchardef\Phig="7C08 \mathchardef\Gammag="7C00
\mathchardef\Psig="7C09 \mathchardef\Lambdag="7C03
\mathchardef\Xig="7C04 \mathchardef\Pig="7C05
\mathchardef\Upsilong="7C07

\begin{document}

\title{Foliations and automorphisms of causal sets}

\author{Ali Bleybel \& Abdallah Zaiour}

\address{Faculty of Sciences (I), Lebanese University, Beirut, Lebanon}





%
\begin{abstract}
We consider foliations of causal sets and their behavior under causal automorphisms. For a connected, countably infinite causal set containing an infinite antichain, and such that every antichain has finite intersection with the past and future of any point, we prove that each automorphism admits a foliation whose slices it either preserves or translates. In the latter case, the slices admit a rational coordinate which the automorphism changes by one unit, with a common choice of direction. The existence of foliations for arbitrary partially ordered sets follows from a theorem of Milner and Pouzet. Here the slices are ordered by finite paths of links between slices. We distinguish this condition from direct linkage and from the additional condition defining a temporal foliation. Counterexamples show that linked foliations need not exist, and that passage to relation space does not ensure temporal existence.

\end{abstract}
%
\maketitle

\section*{Introduction}

One of the important questions concerning the structure of Lorentzian manifolds is whether spacetime can be split into a family of spacelike hypersurfaces. Geroch's theorem~\cite{[G]} relates this question to global hyperbolicity and the existence of a Cauchy surface. The smooth splitting theorem, and the extension of a prescribed smooth spacelike Cauchy surface to a foliation, were established by Bernal and S\'anchez~\cite{[B2],[B3]}. These results motivate the corresponding questions for causal sets.

A causal set, or causet, is a locally finite partially ordered set. The order represents the causal relation, while local finiteness expresses discreteness. In this setting, spacelike slices are antichains. However, an ordered partition into antichains need not have the properties of a foliation by Cauchy hypersurfaces. We shall distinguish a foliation, in which the order of the slices is generated by finite paths of links, from a linked foliation, in which every ordered pair of slices has a direct link. A temporal foliation satisfies an additional condition concerning inextendible chains. The stronger conditions do not hold automatically.

The existence of foliations in the first sense is a consequence of the theorem on antichain decompositions of Milner and Pouzet~\cite[Theorem~2.1]{[MP]}. In Theorem~\ref{foliation} we give the short deduction which also places a prescribed antichain inside one slice. This result provides the order-theoretic starting point for the study of automorphisms.

Our main concern is the existence of a foliation adapted to a causal automorphism. We consider a class of countably infinite causal sets whose comparability graph is connected, which contain an infinite antichain, and for which every antichain meets the past and future of any point in finitely many points. Under these assumptions, Theorem~\ref{auto} gives a foliation on which the automorphism either preserves every slice or translates the slices in a common direction. In the moving case, a rational coordinate on the slices changes by one unit. Without these assumptions, the orbit types still give invariant parts, but the moving-foliation conclusion requires further hypotheses. The Lorentz-boost example below illustrates this distinction.

The paper is organized as follows. Section one fixes our notation. Section two discusses foliations and the passage to relation space. Section three gives the automorphism theorem and its proof, followed by the conclusions.

This version corrects the earlier paper. The direct-link requirement is retained under the name ``linked foliation''. The general temporal-extension assertion and the relation-space temporal-existence assertion are withdrawn. Theorem~\ref{auto} retains the original alternatives, with a corrected proof using Zorn's Lemma. The unrestricted moving and connected-regrouping conclusions formerly stated in Theorem~\ref{general} are not retained. We also withdraw the earlier claims that the FRW volume estimate establishes the required causal-set hypotheses and that the proposed counting formula recovers spatial dimension.

{\it Acknowledgement}: We wish to thank Aron Wall for attracting our attention to some issues with a previous version of the manuscript.

\section{Causal spaces}

\subsection{Causal relations in the continuum}
Let $(\cM,g)$ be a time-oriented Lorentzian manifold of signature $(-,+,\ldots,+)$. The chronological future $I^+(p)$ consists of the points reached from $p$ by future-directed timelike curves; the causal future $J^+(p)$ is defined using causal curves. The pasts $I^-(p)$ and $J^-(p)$ are defined similarly. We write $q\ll p$ for $q\in I^-(p)$, and $q\prec p$ for $q\in J^-(p)$. In a causal spacetime, the latter relation is a partial order. The spacetime is globally hyperbolic if it is causal and every causal diamond $J^+(p)\cap J^-(q)$ is compact. A Cauchy surface meets every inextendible causal curve exactly once.

\subsection{Causal spaces}
In the following, a {\it causal space} $(\cM,\prec)$ is a nonempty set of events equipped with a partial order. We retain the reflexive convention $x\prec x$, and write $x\sprec y$ when $x\prec y$ and $x\ne y$. Two points are {\it incomparable}, written $x\,\|\,y$, if neither $x\prec y$ nor $y\prec x$.

For $x\in\cM$, put
$$ {\rm Past}(x)=\{y\in\cM:y\prec x\},\qquad
   {\rm Fut}(x)=\{y\in\cM:x\prec y\}. $$
The interval with endpoints $x,y$ is
$$[x,y]=\{z\in\cM:x\prec z\prec y\}
       ={\rm Fut}(x)\cap{\rm Past}(y).$$
A {\it causal set} is a locally finite causal space: every interval $[x,y]$ is finite. No restriction on cardinality is imposed unless stated explicitly. A subspace carries the induced order, which we denote by the same symbol.

Our use of ``causal space'' refers only to the partial order. It differs from the more elaborate notion involving causality, chronology and horismos in~\cite{[KP]}.

\section{Temporal foliation}

 Before giving the definition of {\it Temporal foliation}, let us recall some notation.   \\
   By a {\it chain} in a causal space (or a poset) $(\cM, \prec)$ we mean a subset $C$ linearly ordered by $\prec$. A chain is {\it inextendible} if it is maximal under inclusion among chains. \\
 Similarly, by an {\it antichain} we mean a subset $A$ of $\cM$ such that any two distinct elements of $A$ are incomparable, i.e., $\forall x, y \in A, x \neq y \Rightarrow x \, \| \, y$.

 Given an element $x$ in a poset $\cM$ and a subset $X$ of $\cM$, $x$ is said to be {\it incomparable} to $X$ if and only if for every element $y \in X$ we have $x \, \| \, y$, i.e. $x$ and $y$ are incomparable.
 $x$ and $X$ are said to be {\it comparable} otherwise, i.e. there exists at least one $y \in X$ such that $x$ and $y$ are comparable.

 For two disjoint antichains $S,T$, we write $S\longrightarrow T$ if there exist $x\in S$, $y\in T$ with $x\sprec y$, and call this a {\it direct slice link}. A finite path of such links may use different witnessing points in an intermediate slice. Thus a path of slices need not come from a chain of events. We use these finite paths in the definition below, and retain the earlier pairwise condition under the name {\it linked foliation}.

\begin{definition} \label{foliation_0}
  Let $\cM$ be a nonempty partially ordered set. A {\it foliation} $F$ of $\cM$ is a
  partition of $\cM$ into disjoint nonempty sets $X_i, i\in I$ (called spacelike slices), where $I$ is a nonempty set, equipped with a total order $\leq$, such that: 

   \begin{itemize}
    \item[i)] $\cM = \bigcup_{i \in I} X_i$, i.e. the $X_i$'s form a covering of $\cM$. 

    \item[ii)] For every $i \in I$, $X_i$ is an antichain. 

    \item[iii)] For all $i,j\in I$ with $i<j$, there is a finite path
    $$X_i=X_{i_0}\longrightarrow X_{i_1}\longrightarrow\cdots\longrightarrow X_{i_r}=X_j,
      \qquad r\geq1,$$
    and there are no $z\in X_j$, $t\in X_i$ with $z\sprec t$.
   Here $i<j$ means $i \leq j \; \text{and} \; i \neq j$. 

   \end{itemize}

   A foliation is {\it linked} if, for every $i<j$, there exist $x\in X_i$, $y\in X_j$ such that $x\sprec y$. This is the stronger direct-link condition of the original definition. 

   A slice $X_i$ is called a {\it top slice} of $F$ if $i$ is the maximal element of $I$. If $i$ is the minimum element of $I$, then we say that $X_i$ is the {\it bottom slice} of $F$.
   A {\it temporal foliation} $F$ of $\cM$ is a foliation which satisfies furthermore the following requirement: 

   \begin{itemize}
     \item[iv)] Let $X_i$ be a given spacelike slice, and let $C$ be any inextendible chain in $\cM$. Assume $i$ is neither a minimal nor a maximal element of $I$. If $C \cap X_i = \varnothing$ then there exist $x_0, y_0 \in C$ such that $x_0 \sprec z$ and $t \sprec y_0$ for some $z,t \in X_i$. 

   If $X_i$ is a slice of $F$ such that $i$ is a minimal or maximal element of $I$, then any inextendible chain $\gamma$ of $\cM$ intersects $X_i$. 
   \end{itemize}

   The slices of a temporal foliation are called {\it Cauchy slices}. Impermeability alone is not sufficient for the temporal condition. 

   A {\it partial foliation} $F_p$ (resp. {\it partial temporal foliation}) is a foliation (resp. temporal foliation) of some subspace $\cM(F_p) \subseteq \cM$. Here $\cM(F_p)$ is the subspace of $\cM$ consisting of elements contained in slices of $F_p$.
 \end{definition}

  Condition (iii) says that each $X_i$ precedes all $X_j$ with $i<j$ by a finite path of slice links. No direct link points backwards, so the order of the slices is precisely the reachability order of this graph.
  In the case of Minkowski spacetime, one might take $I=\RR$,
   and the $X_i$ are then isomorphic to $\RR^3$.
  If $\cM$ is a globally hyperbolic Lorentzian manifold, a foliation by smooth spacelike Cauchy hypersurfaces satisfies all the above axioms (i)--(iv), and is linked. In this case $\cM$ is homeomorphic to $I\times\Sigma$, where $\Sigma$ is a Cauchy surface.

\begin{figure}[!htbp]
\centering
\includegraphics[width=0.96\textwidth]{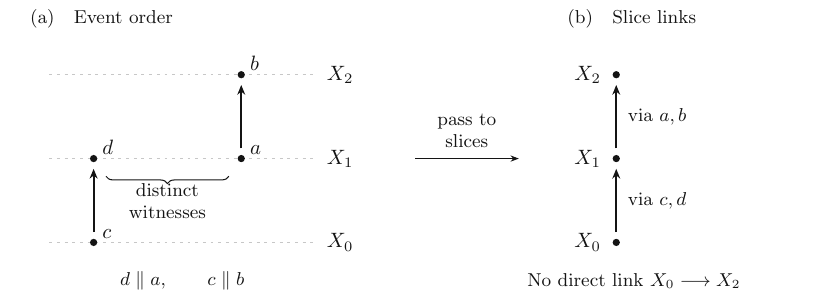}
\caption{Event order and slice links. In (a), the only strict comparisons are $c\sprec d$ and $a\sprec b$. Passing to the slices gives the path $X_0\longrightarrow X_1\longrightarrow X_2$ in (b), but its two links use distinct witnesses $d,a\in X_1$. There is no direct link from $X_0$ to $X_2$. This foliation is neither linked nor temporal.}
\label{fig:slice-paths}
\end{figure}

 One might observe that the {\it slices} of a foliation in our sense do not correspond exactly to the usual notion of a {\it Cauchy hypersurface} in continuous spacetime. A point of a causal set may be incomparable to all events in a chosen slice. Even a maximal antichain can be passed without being met by a chain. For instance, take four points $x,y,a,b$, with the only strict comparisons
 $$x\sprec a,\qquad b\sprec y,\qquad x\sprec y.$$
 Then $\{a,b\}$ is a maximal antichain, whereas the inextendible chain $\{x,y\}$ passes from its past to its future without meeting it. \\
 The relation-space construction considered in \cite{[D2]} addresses this permeability issue. We return to it below, correcting its order relation and distinguishing the existence of impermeable antichains from the existence of a temporal foliation. \\
 Another issue with the proposed analogy between maximal antichains and Cauchy hypersurfaces concerns chains which neither meet an antichain nor have points on both sides of it. The additional axiom (iv) in Definition~\ref{foliation_0} rules out this behavior for internal slices and imposes intersection at endpoint slices.

\begin{remark*}
The distinction between a foliation and a linked foliation is necessary. Let $\cM$ be the disjoint union of the two chains
 $$a_0\sprec a_1\sprec a_2\sprec\cdots,
   \qquad\cdots\sprec b_2\sprec b_1\sprec b_0,$$
with no comparisons between the chains. Suppose that it had a linked foliation. Each slice contains at most one point from each chain. If every slice contained one point of each kind, ordering the slices by increasing $a$-index would force their $b$-indices to decrease, which is impossible. Thus some slice is a singleton.

Suppose this slice is $\{a_i\}$. Full direct linkage implies that only finitely many slices lie below it, since each such slice must contain one of $a_0,\ldots,a_{i-1}$. Infinitely many distinct slices containing $b$-points therefore lie above it, and each must also contain some $a_k$ with $k>i$. As the $b$-indices increase, these slices decrease in the slice order, so their $a$-indices form an infinite strictly decreasing sequence of natural numbers. This is a contradiction. The case of a singleton $\{b_j\}$ is dual. Hence this causal set has no linked foliation.
\end{remark*}

We shall use the following standard form of the axiom of choice:

\begin{theorem} \label{Zorn}
 ({\it Zorn's Lemma})
Let $X$ be any partially ordered set. Assume that every chain has an upper bound. Then $X$ has a maximal element.
\end{theorem}
This result is known to be equivalent to the axiom of choice, which we assume.

Milner and Pouzet~\cite[Theorem~2.1]{[MP]} show that every antichain partition with an acyclic direct-link graph has a coarsening whose reachability order is total. In their terminology, this is the passage from a rankable antichain decomposition to a maximal one. We obtain the following consequence.

  \begin{theorem} \label{foliation}
   Let $(\cM,\prec)$ be a causal space, and let $\Sigma_0$ be an antichain in $\cM$.
   Then there exists a foliation $F$ of $\cM$ such that \\
   $(*)$ There exists a slice $\Sigma\in F$ such that $\Sigma_0\subseteq\Sigma$.
  \end{theorem}
  \begin{proof}
If $\Sigma_0$ is empty, we may replace it by a singleton. Consider the partition
$$F_0=\{\Sigma_0\}\cup
       \bigl\{\{x\}:x\in\cM\setminus\Sigma_0\bigr\}.$$
Its direct-link graph is acyclic. Indeed, a cycle avoiding $\Sigma_0$ would be a strict order cycle in $\cM$. A simple cycle through $\Sigma_0$ would give
$$a\sprec x_1\sprec\cdots\sprec x_r\sprec a',
  \qquad a,a'\in\Sigma_0,$$
contrary to the antichain condition if $a\ne a'$, and to irreflexivity if $a=a'$.

The cited theorem therefore gives a coarsening $F$ of $F_0$ whose reachability order is total. With this order, $F$ is a foliation in the sense of Definition~\ref{foliation_0}. Since $F$ is a coarsening, one of its slices contains $\Sigma_0$.
  \end{proof}

\begin{remark*}
The additional temporal-extension conclusion in the earlier statement of Theorem~\ref{foliation} is withdrawn. Even an antichain meeting every inextendible chain need not be contained in a temporal foliation. To see this, let $\cM=\{a,b,c\}$ with the only strict comparison $a\sprec b$. The antichain $\{a,c\}$ meets both inextendible chains $\{a,b\}$ and $\{c\}$. Any foliation has at least two slices because $a$ and $b$ are comparable. Since the poset is finite, it has distinct bottom and top slices. If the foliation were temporal, both slices would have to meet the chain $\{c\}$, contradicting their disjointness. Hence the existence of such an antichain does not suffice for temporal existence.
\end{remark*}

 \subsection{Passage to the relation space}

  The problem of permeability of spacelike slices motivates passing to the {\it relation space} of a multidirected set, as in \cite{[D2]}. We retain this construction, but must distinguish immediate composability from its transitive closure and correct the earlier temporal-existence assertion. \\
  Recall that a multidirected set is a set $X$ such that between two points there may be several arrows, or none, with the prescribed source and target. A partially ordered set $(X,\prec)$ gives such a directed graph by taking one arrow for each strict comparison $x\sprec y$. In this case there is at most one arrow with given source and target. \\
  Recall that a {\it permeable} spacelike slice is an antichain $A$ such that for some $x,y\in\cM$ there exist distinct $x_0,y_0\in A$ with
  $$x\sprec x_0,\qquad y_0\sprec y,\qquad x\sprec y,$$
  but there is no $z\in A$ with $x\sprec z\sprec y$. This is equivalent to the earlier formulation using the reflexive order throughout: if a potential endpoint belonged to $A$, it would itself provide an intermediary in the closed interval, and the antichain condition makes the witnessing comparisons strict. An antichain is {\it impermeable} if it is not permeable. \\
  The relation space $\cR(\cM)$ is defined as the set of all arrows $a:s\to t$ with $s\sprec t$. We denote their source and target by $s(a):=s$ and $t(a):=t$.

  The relation $a\lhd b$ defined by $t(a)=s(b)$ expresses immediate composability. It is irreflexive but need not be transitive: if $x\sprec y\sprec z\sprec w$, then
  $$(x,y)\lhd(y,z)\lhd(z,w),\qquad
       (x,y)\not\lhd(z,w).$$
  We therefore equip relation space with the reachability relation
  $$a\lhd^*b\quad\Longleftrightarrow\quad t(a)\prec s(b).$$
  Equality of these endpoints is included because $\prec$ is the reflexive order. When applying Definition~\ref{foliation_0} and Theorem~\ref{foliation} to relation space, we use the reflexive closure of $\lhd^*$; all comparisons of distinct arrows remain as stated.

\begin{lemma*}
The relation $\lhd^*$ is the transitive closure of $\lhd$ and is a strict partial order on $\cR(\cM)$.
\end{lemma*}
\begin{proof}
A path of immediately composable arrows from $a$ to $b$ implies $t(a)\prec s(b)$. Conversely, if the endpoints are equal then $a\lhd b$. If $t(a)\sprec s(b)$, the intermediate arrow $(t(a),s(b))$ gives a two-link path from $a$ to $b$.

Since $s(a)\sprec t(a)$, one cannot have $t(a)\prec s(a)$, so $\lhd^*$ is irreflexive. If $a\lhd^*b\lhd^*c$, then
$$t(a)\prec s(b)\sprec t(b)\prec s(c),$$
and hence $a\lhd^*c$. Thus it is transitive as well.
\end{proof}

  The use of maximal antichains to address impermeability is discussed in \cite{[D2]}. For the corrected reachability order, we give a direct proof of the following assertion.
   \begin{theorem}
   Maximal antichains in relation space, equipped with $\lhd^*$, are impermeable.
   \end{theorem}
\begin{proof}
Let $A$ be a maximal antichain. Suppose that arrows $a\lhd^*b$ and distinct $c,d\in A$ satisfy $a\lhd^*c$ and $d\lhd^*b$. Put $u=t(a)$ and $v=s(b)$, so $u\prec v$. If $u=v$, then
$$t(d)\prec v=u\prec s(c),$$
which gives $d\lhd^*c$, contrary to antichainness. Hence $u\sprec v$, and $e=(u,v)$ is an arrow in relation space.

If $e\lhd^*f$ for some $f\in A$, then $t(d)\prec v\prec s(f)$ and therefore $d\lhd^*f$, which is impossible. If $f\lhd^*e$, then $t(f)\prec u\prec s(c)$ and therefore $f\lhd^*c$, again impossible. Thus $e$ is incomparable with every distinct element of $A$. Maximality forces $e\in A$. By construction,
$$a\lhd^*e\lhd^*b.$$
Every proposed permeating pair therefore has an intermediary in $A$, proving impermeability. The comparisons here are strict comparisons of arrows, so neither endpoint is being used as the intermediary.
\end{proof}

   Moreover, relation space is locally finite if $\cM$ is locally finite. Indeed, for $a\lhd^*b$, any arrow $e$ strictly between them has both endpoints in the finite interval $[t(a),s(b)]$ of $\cM$. There are only finitely many such arrows, and adjoining the two endpoints $a,b$ still gives a finite interval.

   This result concerns individual antichains. It does not imply the existence of a temporal foliation into such antichains. The earlier assertion at Theorem~\ref{impermeable} is withdrawn and replaced by the following counterexample.
   \begin{theorem} \label{impermeable}
   There is a finite causal set whose relation space admits no temporal foliation, for either immediate composability or its transitive closure.
   \end{theorem}
   \begin{proof}
   Let $\cM=\{0\sprec1\sprec2\}$. Its relation space consists of
   $$\alpha=(0,1),\qquad\beta=(1,2),\qquad\gamma=(0,2).$$
   The only strict comparison is $\alpha\lhd^*\beta$, and immediate composability already has this same relation. The arrow $\gamma$ is incomparable with the other two, so $\{\gamma\}$ is an inextendible chain.

   Suppose there were a temporal foliation. The comparable arrows $\alpha,\beta$ must lie in distinct slices. Since the relation space is finite, the foliation has distinct bottom and top slices. Both must meet the chain $\{\gamma\}$ by axiom (iv), so both would contain $\gamma$, contrary to disjointness.
   \end{proof}
   In this example, both $\{\alpha,\gamma\}$ and $\{\beta,\gamma\}$ are maximal, impermeable antichains meeting every inextendible chain. They cannot be distinct slices of a partition because they share $\gamma$. Theorem~\ref{foliation} still gives a foliation of every nonempty relation space with the corrected order, but makes no temporal assertion. The counterexample does not decide the weaker question of a foliation into impermeable slices without the temporal condition.
\section{Automorphisms of infinite causal sets}
\subsection{Some assumptions} \label{assump}
Let $\cC$ be a countably infinite causal set. We call $\cC$ {\it special} if it satisfies the following assumptions:
\begin{itemize}
\item[(a)] The undirected comparability graph of $\cC$ is connected.
\item[(b)] For every $x\in\cC$ and every antichain $A$, both $ {\rm Past}(x)\cap A$ and $ {\rm Fut}(x)\cap A$ are finite.
\item[(c)] $\cC$ contains an infinite antichain.
\end{itemize}
These assumptions are restrictive. The following examples show that they hold in a simple lattice, but need not hold for other discrete spacetime models.

\subsubsection{Example} \label{aron}

Take $\cC=\mathbb Z^2$ with the coordinate order:
$$ (m,n)\prec(m',n')\quad\Longleftrightarrow\quad m\leq m'\ \text{and}\ n\leq n'. $$
Intervals are finite rectangles, and horizontal and vertical steps show that the comparability graph is connected. The set $\{(n,-n):n\in\mathbb Z\}$ is an infinite antichain. Finally, a future is a translate of $\mathbb N_0^2$. In an antichain in this set, increasing first coordinates would force strictly decreasing second coordinates; equal first coordinates cannot occur for distinct points. Hence such an antichain is finite. Reversing the order gives the same conclusion for a past. Thus $\mathbb Z^2$ is a special causal set.

\subsubsection{A counterexample to assumption (b)} \label{boost-lattice}
 Let us consider~\cite{[W]} a $1+1$ flat Lorentzian spacetime lattice $\cL$, generated by the vectors $u=(1,0)$ (timelike) and $v=(1/2,\sqrt5/2)$ (spacelike). Let $\Phi_L:\cL\to\cL$ be the Lorentz boost sending $(-1/2,\sqrt5/2)$ to $(1/2,\sqrt5/2)$. Its matrix is
$$ B=\begin{pmatrix}3/2&\sqrt5/2\\ \sqrt5/2&3/2\end{pmatrix}. $$
Since $Bu=u+v$ and $Bv=u+2v$, the matrix in the lattice basis is integral with determinant one. Thus $\Phi_L$ is a bijection of $\cL$ and preserves its causal order.
The lattice is locally finite, since a bounded causal diamond contains only finitely many lattice points. Its comparability graph is connected, since the timelike vectors $u$ and $u+v$ generate the lattice.

The lattice point $(0,0)$ is in the past of the infinite antichain $A$ defined as the orbit of $(1,0)$. Indeed, this orbit lies on the future hyperbola $t^2-x^2=1$, and two distinct points on that hyperbola are spacelike separated. Hence assumption (b) fails.

\subsection{Automorphisms}
Let $\cC$ be a special causal set. A {\it causal automorphism} is a bijection $\Phi:\cC\to\cC$ such that
$$x\prec y\quad\Longleftrightarrow\quad\Phi(x)\prec\Phi(y).$$
Its inverse and all its integer powers are again causal automorphisms. The {\it orbit} of $y\in\cC$ is
$$\cO(y)=\{\Phi^\ell(y):\ell\in\mathbb Z\}.$$
Here $\Phi^0$ is the identity, and negative powers are iterates of $\Phi^{-1}$.

 We have the following results:
\begin{prop}\label{prop-anti_chain}
Let $\Phi:\cC\to\cC$ be a causal automorphism, $\cC$ being an arbitrary causal set. Let $z\in\cC$, and assume that $z$ and $\Phi^\ell(z)$ are incomparable or equal for all $\ell\in\mathbb Z\setminus\{0\}$. Then the orbit $\cO(z)$ is an antichain.
\end{prop}
\begin{proof}
Let $z$ be as above. To see that any two distinct orbit points are incomparable, we take the contrapositive. If, say, $\Phi^\ell(z)\sprec\Phi^k(z)$, applying $\Phi^{-k}$ to both sides gives $\Phi^{\ell-k}(z)\sprec z$. This contradicts the hypothesis that $z$ and this iterate are incomparable or equal. The opposite comparison is handled in the same way.
\end{proof}

\begin{lem}\label{lem}
Let $\cC$ be a causal set satisfying the assumptions (a)--(c) above and let $\Phi$ be a non-trivial automorphism of $\cC$. Then there are three mutually incompatible possibilities for $\Phi$: \\
(i) For all $x\in\cC$, there exists $k\in\mathbb N$, $k\neq0$, such that $x\sprec\Phi^k(x)$; \\
(ii) For all $x\in\cC$, there exists $k\in\mathbb N$, $k\neq0$, such that $\Phi^k(x)\sprec x$; \\
(iii) For all $x\in\cC$, the orbit $\cO(x)$ is an antichain.
\end{lem}
\begin{proof}
Let $\Phi$ be a non-trivial automorphism. We first observe that a point cannot satisfy both of the moving conditions: taking a common positive multiple of the two exponents would give a strict order cycle. If neither condition holds, Proposition~\ref{prop-anti_chain} shows that its orbit is an antichain. A finite orbit is necessarily an antichain, since iterating a strict comparison between two of its points would also give a strict order cycle.

We show the following claim, including the possibility of a periodic orbit. \\
{\bf Claim:} If $\Phi$ has a finite orbit, then the orbit of every element of $\cC$ is a finite antichain.

To see this, let $x_0$ have a finite orbit. By the connectedness of $\cC$, for any $x\in\cC$ there is a finite comparability path from $x_0$ to $x$. We prove the claim along this path. It suffices to show that if $v$ has a finite orbit, then any point $u$ comparable with $v$ also has a finite orbit.

Choose $q>0$ such that $\Phi^q(v)=v$, and suppose first that $u\sprec v$. If the orbit of $u$ is infinite and has a strict comparison, choose $m>0$ so that either
$u\sprec\Phi^m(u)$ or $\Phi^m(u)\sprec u$.
In the first case,
$$u\sprec\Phi^{jmq}(u)\sprec v\qquad(j>0),$$
and in the second case the same conclusion holds with the exponent $-jmq$. Either way there are infinitely many distinct points in $[u,v]$, contrary to local finiteness.
The remaining possibility is that $\cO(u)$ is an infinite antichain. But
$$\Phi^{jq}(u)\sprec\Phi^{jq}(v)=v\qquad(j\in\mathbb Z)$$
puts infinitely many of its points in $ {\rm Past}(v)$, contrary to assumption (b).
The case $v\sprec u$ follows by reversing the order. Thus finite orbits propagate along the path, proving the claim. In particular, the claim applies when $\Phi$ has a fixed point.

For the remainder of the proof we may therefore suppose that all orbits are infinite. Assume that, for some $x_0\in\cC$, we have
$x_0\sprec\Phi^k(x_0)$ for some $k>0$.
Let $x\in\cC$. By assumption (a), there is a path
$x_0,x_1,\ldots,x_n=x$ whose consecutive points are comparable.
We shall show that the same moving condition passes from one point of this path to the next.

Suppose that $u\sprec\Phi^k(u)$ and that $w$ is comparable with $u$.
If $w$ satisfies the opposite moving condition, choose $\ell>0$ with
$\Phi^\ell(w)\sprec w$, and put $m_0=k\ell$.
When $u\sprec w$, we obtain
$$u\sprec\Phi^{jm_0}(u)\sprec\Phi^{jm_0}(w)\sprec w
 \qquad(j>0).$$
This contradicts local finiteness of $[u,w]$. When $w\sprec u$, use negative iterates instead:
$$w\sprec\Phi^{-jm_0}(w)\sprec\Phi^{-jm_0}(u)\sprec u
 \qquad(j>0).$$
This gives the same contradiction in $[w,u]$.

If instead $\cO(w)$ is an antichain, it is infinite by the reduction above.
For $u\sprec w$, the comparisons
$$u\sprec\Phi^{jk}(u)\sprec\Phi^{jk}(w)\qquad(j>0)$$
put infinitely many of its points in $ {\rm Fut}(u)$. For $w\sprec u$, the comparisons
$$\Phi^{-jk}(w)\sprec\Phi^{-jk}(u)\sprec u\qquad(j>0)$$
put infinitely many of its points in $ {\rm Past}(u)$.
Both possibilities contradict assumption (b).
Thus $w\sprec\Phi^\ell(w)$ for some $\ell>0$. Induction along the path proves (i).
Applying the same argument to $\Phi^{-1}$ proves (ii) whenever a point has that moving type. If no point has either moving type, (iii) holds. The pointwise alternatives are mutually incompatible, so the three global alternatives are also mutually incompatible.
\end{proof}

The main result of this section is the following:
\begin{theorem}\label{auto}
Let $\cC$ and $\Phi$ be as above, with $\Phi$ distinct from the identity. \\
There exists a foliation $F$ of $\cC$ such that either: \\
(a) For every spacelike slice $X$ of $F$, $\Phi(X)=X$ ($\Phi$ restricts to an automorphism of $X$, for each slice $X$ of $F$); or \\
(b) There exists a nonempty chain $C\subset\cC$ and an integer $k\neq0$ such that, for every $x\in C$, $x\sprec\Phi^k(x)$, and for every spacelike slice $X$ of $F$ there exists a spacelike slice $Y$ of $F$ such that $\Phi(X)=Y$ and $X\neq Y$.
\end{theorem}
\begin{proof}
We consider the following three cases: \\
Case 0---$\Phi$ has a finite orbit, including the possibility of a fixed point. \\
Case 1---$\Phi$ has no finite orbit, and for some $x\in\cC$ we have
$x\sprec\Phi(x)$ or $\Phi(x)\sprec x$. \\
Case 2---$\Phi$ has no finite orbit, and $\Phi(x)$ and $x$ are incomparable for every $x\in\cC$.

Considering Case 0, it follows by Lemma~\ref{lem} that every orbit is an antichain. We shall find a foliation whose slices are unions of whole orbits.
Let $Q$ be the set of cyclic orbits, and declare one orbit to precede another when some point of the first strictly precedes a point of the second.
This is a strict partial order. Indeed, when two comparisons use different representatives of their common middle orbit, we apply an iterate of $\Phi$ to the first comparison to make these representatives agree. Transitivity then follows. A comparison of an orbit with itself would contradict the fact that it is an antichain.
Take the reflexive closure to regard $Q$ as a poset, apply Theorem~\ref{foliation}, and pull its slices back to $\cC$.
Each pullback is an antichain, and the direct slice links correspond exactly to those in $Q$. Hence reachability is total and acyclic, so the pullbacks form a foliation.
Every slice is a union of orbits and is preserved by $\Phi$. Thus (a) holds. This construction applies whenever all orbits are antichains, whether finite or infinite.

Let us now consider Case 1. Choose $x_0$ with $x_0\sprec\Phi(x_0)$, the opposite case being obtained by replacing $\Phi$ with $\Phi^{-1}$.
By Lemma~\ref{lem}, every point satisfies condition (i).
The orbit of $x_0$ is a chain, and its singleton slices form a partial foliation.
We assign to these slices the heights
$$h_0(\Phi^n(x_0))=n\qquad(n\in\mathbb Z).$$

We shall use a maximal-extension construction, keeping track of the heights so that an entire orbit can be inserted at once.
Let $\cE$ be the set of partial foliations $F$ with $\Phi$-invariant domain $D\subseteq\cC$, together with a map $h:D\to\mathbb Q$ extending $h_0$, such that:
\begin{itemize}
\item the nonempty fibers $\Sigma_t=h^{-1}(t)$ are the slices of $F$;
\item $h(\Phi x)=h(x)+1$ for every $x\in D$;
\item $x\sprec y$ in $D$ implies $h(x)<h(y)$;
\item whenever $s<t$ are attained heights, a finite path of direct slice links leads from $\Sigma_s$ to $\Sigma_t$.
\end{itemize}
In particular,
$$\Phi(\Sigma_t)=\Sigma_{t+1}\neq\Sigma_t.$$
Thus the following two conditions hold: \\
$(*)$ For every spacelike slice $X\in F$ and every $x\in X$, $\Phi(x)\notin X$. \\
$(\dagger)$ For every spacelike slice $X\in F$, there is a distinct slice $Y\in F$ with $\Phi(X)=Y$.
The initial singleton foliation belongs to $\cE$, so $\cE$ is nonempty.

Define $(F,h)\sqsubseteq(F',h')$ when $D\subseteq D'$ and $h'|_D=h$.
Every old slice is then contained in the new slice of the same height.
The empty chain is bounded above by the seed. Let $\cY$ be a nonempty chain for this order. Take the union of its domains and of its compatible height maps.
Any two points belong to a common member of the chain, so strict order preservation passes to the union. The fibers are therefore antichains, and every direct link increases height.
Given two attained heights, both occur in some common member of the chain. The finite path linking their fibers there persists in the union, since the heights of all witnessing points remain unchanged.
Consequently the union fibers have total, acyclic reachability and form a partial foliation.
Equivariance also passes to the union, and gives $\Phi(\Sigma_t)=\Sigma_{t+1}$.
We have thus obtained an upper bound in $\cE$.
Applying Zorn's Lemma, let $(F_{\rm max},h)$ be a maximal element.

We claim that its domain $D$ is all of $\cC$.
Suppose otherwise and let $x\notin D$.
If the orbit of $x$ met a slice of $F_{\rm max}$, invariance under all integral iterates would give $x\in D$, a contradiction.
Thus the whole orbit is missing. Moreover, we may choose a missing orbit adjacent to $D$: take a comparability path from a point of $D$ to a point outside $D$, and use the first point of the path outside $D$.
We choose its representative $x$ comparable with a point of $D$.

Before inserting this orbit, let us justify the extrema of the two cuts.
Put
$$B^-(x)=\{h(y):y\in D,\ y\sprec x\},\qquad
  B^+(x)=\{h(y):y\in D,\ x\sprec y\}.$$
{\it Claim:} Every nonempty lower cut has an attained maximum, and every nonempty upper cut has an attained minimum.

To prove the claim, consider $A=\{y\in D:y\sprec x\}$.
Given $y\in A$, the set $A\cap[y,x]$ is finite and nonempty. A maximal member of this set is maximal in $A$, so every point of $A$ lies below a maximal member of $A$.
The maximal members form an antichain contained in $ {\rm Past}(x)$; by assumption (b), there are only finitely many.
Since $h$ preserves strict comparisons, the largest of their heights is the maximum of $B^-(x)$.
The upper cut is treated in the same way using minimal members of $D\cap{\rm Fut}(x)$.
This proves the claim. Notice that it applies to the whole domain $D$, which need not contain only finitely many orbits.

Write $L=\max B^-(x)$ when the lower cut is nonempty, and $L=-\infty$ otherwise.
Write $U=\min B^+(x)$ when the upper cut is nonempty, and $U=+\infty$ otherwise.
At least one is finite, by our choice of $x$.
If both are finite, witnesses $p\sprec x\sprec q$ with $h(p)=L$ and $h(q)=U$ give $L<U$.
Equivariance shows that the corresponding extrema for $\Phi^n(x)$ are $L+n$ and $U+n$, whenever they exist.
Finally, put $H=h(D)$; then $H+1=H$.

There are two subcases.

(b-i) Suppose that some $t\in H$ satisfies $L<t<U$.
This includes every one-sided cut: take $t=U-1$ when $L=-\infty$, or $t=L+1$ when $U=+\infty$.
Extend the height by
$$h(\Phi^n(x))=t+n\qquad(n\in\mathbb Z).$$
Every comparison from an old point to $\Phi^n(x)$ has its old height at most $L+n<t+n$, and every comparison in the other direction has old height at least $U+n>t+n$; the corresponding direction is absent when its cut is empty.
Comparisons within the new orbit also increase the exponent, since all points satisfy condition (i).
Thus the extended height is strictly order-preserving. No new heights have been introduced, and all the old paths between fibers persist.
The extended fibers therefore form a partial foliation in $\cE$, strictly extending $F_{\rm max}$, a contradiction.

(b-ii) Suppose there is no height of $H$ strictly between $L$ and $U$.
Both cuts are then finite. Choose a rational number $a$ with $L<a<U$ and put
$$h(\Phi^n(x))=a+n\qquad(n\in\mathbb Z).$$
The same cut argument proves strict order preservation.
Choose witnesses $p\sprec x\sprec q$ for $L$ and $U$.
Their translates give direct slice links
$$\Sigma_{L+n}\longrightarrow\Sigma_{a+n}
  \longrightarrow\Sigma_{U+n}\qquad(n\in\mathbb Z).$$
The heights $a+n$ are new, because $H+1=H$ and $H$ has no height between $L$ and $U$.
Also $L+1\in H$, so $U\le L+1$.
If $U<L+1$, the old foliation gives a finite path from $\Sigma_{U+n}$ to $\Sigma_{L+n+1}$; if equality holds, these two fibers coincide.
Together with the displayed links and their translates, this gives a finite path from $\Sigma_{a+n}$ to $\Sigma_{a+n+1}$.
Every old height below $a+n$ is at most $L+n$, and every old height above it is at least $U+n$.
It follows that each new fiber is comparable for reachability with every old fiber, while the consecutive paths compare every pair of new fibers.
Reachability is thus total and acyclic. We again obtain a member of $\cE$ strictly extending $F_{\rm max}$, a contradiction.

In both subcases we contradict maximality. Hence $D=\cC$, and $F_{\rm max}$ is a foliation of $\cC$.
The singleton seed already gives the chain required in (b), with $k=1$; the equivariant height shows that every slice is carried to a distinct slice.
The past-moving case follows by applying the argument to $\Phi^{-1}$.

\begin{figure}[!htbp]
\centering
\includegraphics[width=0.96\textwidth]{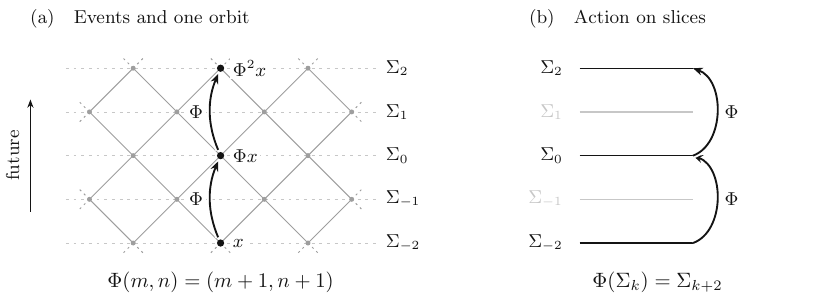}
\caption{A moving automorphism (Case 1). In (a), $\mathbb Z^2$ has the coordinate order and is drawn in coordinates $(m-n,m+n)$. Grey lines are causal covers, directed upwards; curved black arrows show the orbit of $x=(-1,-1)$ under $\Phi(m,n)=(m+1,n+1)$. In (b), the induced action sends each slice $\Sigma_k=\{m+n=k\}$ to $\Sigma_{k+2}$. Hence $h=(m+n)/2$ increases by one. The diagrams show finite portions of the poset and the slice family.}
\label{fig:moving-automorphism}
\end{figure}

Now consider Case 2, where $x$ and $\Phi(x)$ are incomparable for every $x$.
There are two subcases: \\
(2-a) Every cyclic orbit is an antichain. \\
(2-b) Some $x$ is strictly comparable with $\Phi^k(x)$ for an integer $k$ with $|k|>1$.

In Case (2-a), we use the same construction as in Case 0: a foliation of the orbit poset pulls back to a foliation of $\cC$ whose slices are unions of whole orbits. Thus every slice is preserved by $\Phi$, and (a) holds.

\begin{figure}[!htbp]
\centering
\includegraphics[width=0.96\textwidth]{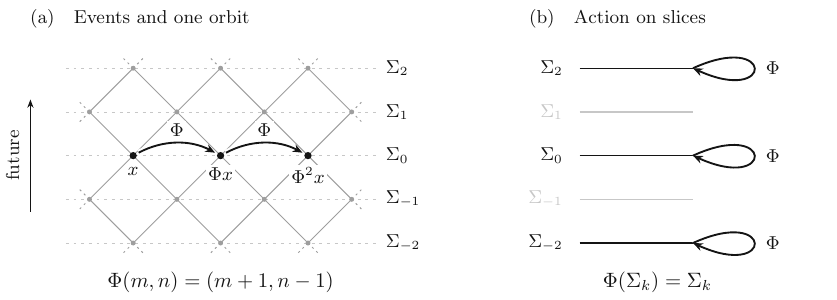}
\caption{A slice-preserving automorphism (Case (2-a)), with the same lattice and conventions as Figure~\ref{fig:moving-automorphism}. In (a), $\Phi(m,n)=(m+1,n-1)$ moves $x=(-1,1)$ within $\Sigma_0$; its arrows join incomparable events. Every orbit is an infinite antichain. In (b), the loops indicate that the slices are preserved setwise.}
\label{fig:spatial-automorphism}
\end{figure}

Next we handle Case (2-b).
Replacing $\Phi$ with its inverse if necessary, let $k_0>1$ be the least positive integer for which some $x_0$ satisfies
$x_0\sprec\Phi^{k_0}(x_0)$.
By Lemma~\ref{lem}, every point satisfies condition (i), but no point satisfies $z\sprec\Phi(z)$.
The set
$$C=\{\Phi^{jk_0}(x_0):j\in\mathbb Z\}$$
is a chain and will give the first assertion of (b).
To apply the preceding maximal-extension argument, it remains to construct an initial partial foliation with an equivariant height.

If $z\sprec\Phi^r(z)$ with $r>0$, each residue class of exponents modulo $r$ is a chain.
An antichain therefore meets $\cO(z)$ in at most $r$ points.
A finite union of moving orbits has finite width, so the infinite antichain required by (c) forces infinitely many orbits.
By connectedness, two distinct orbits contain comparable points.
Choose representatives $x,y$ so that
$$x\sprec\Phi^d(y),$$
where $d$ is the least integer for which this comparison holds.
Such a least integer exists: if the integers $m$ with $x\sprec\Phi^m(y)$ were unbounded below, choose $r>0$ with $y\sprec\Phi^r(y)$ and a decreasing sequence of those integers in one residue class modulo $r$, beginning with $m_0$.
This would put infinitely many distinct points in one interval $[x,\Phi^{m_0}(y)]$, contrary to local finiteness.

On the union of these two orbits, put
$$h_0(\Phi^n(x))=n,\qquad
  h_0(\Phi^n(y))=n+1-d.$$
The fiber at height $n$ is
$$\Sigma_n=\{\Phi^n(x),\Phi^{n+d-1}(y)\}.$$
Its two points are incomparable. Indeed, $x\sprec\Phi^{d-1}(y)$ contradicts the minimality of $d$; the opposite comparison gives
$$x\sprec\Phi^d(y)\sprec\Phi(x),$$
contrary to the assumption of Case 2.

Let us check strict order preservation.
If $\Phi^i(x)\sprec\Phi^j(y)$, then $j-i\ge d$, so the height difference is $j-i+1-d\ge1$.
If $\Phi^i(y)\sprec\Phi^j(x)$, put $e=j-i$. Concatenation with $x\sprec\Phi^d(y)$ gives
$$x\sprec\Phi^{d+e}(x).$$
Condition (i) excludes $d+e<0$, a strict order cycle excludes $d+e=0$, and Case 2 excludes $d+e=1$.
Thus $d+e\ge2$, and the height difference is $d+e-1\ge1$.
Comparisons within either orbit increase the exponent and hence the height.
Finally,
$$\Phi^n(x)\sprec\Phi^{n+d}(y)\in\Sigma_{n+1},$$
so consecutive fibers are directly linked. This proves total reachability between all the integer fibers.
We have obtained an initial partial foliation satisfying all the height conditions used in Case 1.

Let $\cE_3$ now consist of the height-decorated partial foliations extending this two-orbit seed.
It is nonempty by the construction just given.
The chain-union argument and both orbit-insertion subcases from Case 1 apply without change.
Zorn's Lemma therefore gives a maximal member whose domain is all of $\cC$.
Its fibers form the required foliation, and $\Phi(\Sigma_t)=\Sigma_{t+1}$.
Together with the chain $C$ above, this proves (b). The opposite moving direction is treated by inversion.
\end{proof}

\noindent{\it Remark.}
In the moving cases the proof gives a height set $H\subseteq\mathbb Q$ with $H+1=H$ and $\Phi(\Sigma_t)=\Sigma_{t+1}$ or $\Phi(\Sigma_t)=\Sigma_{t-1}$, according to the direction of motion. The set $H$ need not be $\mathbb Z$ or even discrete. This height is an order parameter, not a proper-time coordinate. We do not assert that the adapted foliation is temporal. The theorem concerns the existence of an adapted foliation; it does not prescribe the action on every invariant foliation.
\subsection{General causal sets}
What happens when we drop assumption (b) above?
Let more generally $\cC$ be a causal set satisfying (a) and (c) above, and let $\Phi:\cC\to\cC$ be a causal automorphism.
First, observe that the following holds:
\begin{prop}\label{prop}
Let $\cC$ be a causal set satisfying assumptions (a) and (c) from Section~\ref{assump}, and let $\Phi:\cC\to\cC$ be a causal automorphism distinct from the identity.
There is a partition of $\cC$ into $\Phi$-invariant subcausets $\cC_i$, $i=1,2,3$, allowing some of them to be empty, such that:
\begin{itemize}
\item[i)] For all $x\in\cC_1$, there exists $k\in\mathbb N^*$ such that $x\sprec\Phi^k(x)$;
\item[ii)] For all $x\in\cC_2$, there exists $k\in\mathbb N^*$ such that $\Phi^k(x)\sprec x$;
\item[iii)] For all $x\in\cC_3$, the cyclic orbit of $x$ is an antichain, including the possibility of a finite orbit.
\end{itemize}
\end{prop}
\begin{proof}
The pointwise trichotomy in the proof of Lemma~\ref{lem} shows that each point has exactly one of these properties. Whenever an element has one of them, all elements of its orbit have the same property. Now take each $\cC_i$ as the union of the orbits of elements satisfying the corresponding property.
\end{proof}

Using the above proposition, one obtains the following qualified version of the earlier general statement:
\begin{theorem}\label{general}
Let $\cC$ and $\Phi$ be as in Proposition~\ref{prop}, with the partition given there. If $\cC_3\ne\varnothing$, it admits a foliation whose slices are preserved setwise by $\Phi$.
On any $\Phi$-invariant special subcauset contained in $\cC_1$ or $\cC_2$, the corresponding moving conclusion of Theorem~\ref{auto} holds.
\end{theorem}
\begin{proof}
On $\cC_3$, apply the orbit-poset argument from the spatial case of Theorem~\ref{auto}; this argument requires only that every cyclic orbit be an antichain. On an invariant special subcauset in either moving part, the hypotheses of Theorem~\ref{auto} hold, and we apply that theorem to the induced order.
\end{proof}

The earlier proof did not establish that the moving conclusions hold on all of $\cC_1$ and $\cC_2$ after assumption (b) is removed. Nor did it justify the asserted regrouping into connected subcausets merely by counting the spatial orbits. We do not retain those assertions here. In particular, the pieces of the partition need not themselves satisfy the hypotheses of the moving theorem.

Here we stress that the partition depends upon the automorphism $\Phi$; different automorphisms may give different partitions. The foliations obtained on individual pieces are not asserted to combine into a foliation of $\cC$.

\subsection{Example}
Consider again the Lorentz boost $\Phi_L:\cL\to\cL$ of the lattice spacetime described in Section~\ref{boost-lattice}. Let us partition $\cL$ into three discrete subsets $L_1,L_0,L_2$. By $L_0$ we denote the causal cone of the origin, including both its future and its past, while
$$ L_1=\{(t,x)\in\cL:x>|t|\},\qquad
   L_2=\{(t,x)\in\cL:x<-|t|\}. $$
All three subsets are preserved by the boost. There are no nonzero null points in this lattice: for a point $mu+nv$, the equality $t^2-x^2=0$ would give $m^2+mn-n^2=0$, which has no nonzero integral solution.

On $L_0$, every cyclic orbit is an antichain; the origin is fixed and the other orbits lie on future or past timelike hyperbolae. On $L_1$, the boost moves each point strictly to its future, and on $L_2$ it moves each point strictly to its past. To check the latter assertions, put $(t',x')=B(t,x)$ and note that
$$ (t'-t)^2-(x'-x)^2=x^2-t^2>0,\qquad
   t'-t=\frac{t+\sqrt5\,x}{2}. $$
The sign of the time difference is positive on $L_1$ and negative on $L_2$.
Thus this example exhibits the three orbit types in one connected causal set. It illustrates the invariant partition, without relying on the unrestricted moving-foliation assertion withdrawn above.

\clearpage
\section{Conclusion}
The antichain decomposition theorem of Milner and Pouzet supplies the general existence result used here. Under the additional assumptions of Section~\ref{assump}, we have proved that a causal automorphism admits a foliation on which it either preserves each slice or translates the slices. The hypotheses are essential to the argument: without them, the same automorphism may have different orbit types on different invariant parts.

The existence of such a foliation does not settle the temporal question. The finite counterexamples in Section two show that neither an antichain meeting every inextendible chain nor passage to relation space ensures a temporal foliation. It remains to find useful additional hypotheses under which the foliations adapted to automorphisms satisfy the temporal condition.

\bibliographystyle{amsplain}

\end{document}